\documentclass[12pt]{article}
\usepackage[fleqn]{amsmath}
\usepackage{amsthm}
\usepackage{amssymb,latexsym,color,dsfont}
\usepackage[mathscr]{eucal}
\usepackage[margin=1.1in]{geometry}
\usepackage[auth-sc]{authblk}

\newtheorem{theorem}{Theorem}[section]
\newtheorem{lemma}[theorem]{Lemma}
\newtheorem{sublemma}{Sublemma}

\newtheorem{corollary}[theorem]{Corollary}

\newtheorem{remark}[theorem]{Remark}

\newtheorem*{claim-proof}{Proof of the claim}
\renewenvironment{proof}{\noindent
  \textbf{Proof.}}{\hfill$\Box$\\}

\newcommand{\noop}[1]{}
\newcommand{\cut}[1]{}

\newcommand{\instr}[5]{\ensuremath{\hbox to 60 pt
    {${#1}$\hfil${#2}$\hfil$ \rightarrow
      $\hfil${#3}$\hfil${#4}$\hfil${#5}$}}}

\newcommand{\vp}{\ensuremath{\varphi}}

\newcommand{\nat}{\ensuremath{\mathbb{N}}}

\newcommand{\con}{\wedge}

\newcommand{\dis}{\vee}

\newcommand{\imp}{\rightarrow}

\newcommand{\truth}{\ensuremath{\top}}

\newcommand{\mmodel}[1]{\ensuremath{\frak{#1}}}

\newcommand{\sat}[3]{\ensuremath{\frak{#1}, #2 \models #3}}

\newcommand{\notsat}[3]{\ensuremath{\mmodel{#1}, #2 \not\models #3}}

\newcommand{\psmodel}[1]{\ensuremath{\mathcal{M}^*}}
\newcommand{\pseudomodel}[1]{\ensuremath{\mathcal{M}^{**}}}

\begin{document}

\sloppy

\title{Undecidability \\ of first-order modal and intuitionistic
  logics with \\
  two variables and one monadic predicate letter\thanks{Corrected
    version of the paper published in \textit{Studia Logica}, volume
    107, pages 695--717 (2019). doi:10.1007/s11225-018-9815-7.}}
\author{Mikhail Rybakov\thanks{\texttt{m\_rybakov@mail.ru}} and Dmitry
  Shkatov\thanks{\texttt{shkatov@gmail.ru}}}

\date{\today}
\maketitle

\begin{abstract}
  We prove that the positive fragment of first-order intuitionistic
  logic in the language with two individual variables and a single
  monadic predicate letter, without functional symbols, constants, and
  equality, is undecidable.  This holds true regardless of whether we
  consider semantics with expanding or constant domains.  We then
  generalise this result to intervals $[{\bf QBL}, {\bf QKC}]$ and
  $[{\bf QBL}, {\bf QFL}]$, where {\bf QKC} is the logic of the weak
  law of the excluded middle and {\bf QBL} and {\bf QFL} are
  first-order counterparts of Visser's basic and formal logics,
  respectively.  We also show that, for most ``natural'' first-order
  modal logics, the two-variable fragment with a single monadic
  predicate letter, without functional symbols, constants, and
  equality, is undecidable, regardless of whether we consider
  semantics with expanding or constant domains.  These include all
  sublogics of {\bf QKTB}, {\bf QGL}, and {\bf QGrz}---among them,
  {\bf QK}, {\bf QT}, {\bf QKB}, {\bf QD}, {\bf QK4}, and {\bf QS4}.
\end{abstract}

\section{Introduction}
\label{sec:introduction}

While the (first-order) quantified classical logic {\bf QCl} is
undecidable \cite{Church36}, it contains a number of rather expressive
decidable fragments~\cite{BGG97}.  This has long stimulated interest
in drawing the borderline between decidable and undecidable fragments
of {\bf QCl} using a variety of criteria, in isolation or in
combination, imposed on the language.  One such criterion is the
number and arity of predicate letters allowed in the language: while
the monadic fragment is decidable~\cite{Behmann22}, the fragment
containing a single binary letter is not, as follows
from~\cite{Godel33}.  Another criterion is the number of individual
variables allowed in the language: while the two-variable fragment is
decidable \cite{Mortimer75, GKV97}, the three-variable fragment is not
\cite{Suranyi43}.

Similar questions have long been of interest in (first-order)
quantified intuitionistic and modal logics. For languages without
restrictions on the number of individual variables,
Kripke~\cite{Kripke62} has shown that all ``natural'' quantified modal
logics with two monadic predicate letters are undecidable, while
Maslov, Mints, and Orevkov~\cite{MMO65} and, independently,
Gabbay~\cite{Gabbay81} have shown that quantified intuitionistic logic
with a single monadic predicate letter is undecidable.

The question of where the borderline lies in the intuitionistic and
modal case when it comes to the number of individual variables allowed
in the language has recently been investigated by Kontchakov, Kurucz,
and Zakharyschev in~\cite{KKZ05}.  It is shown in~\cite{KKZ05} that
two-variable fragments of quantified intuitionistic and all
``natural'' modal logics are undecidable.  Moreover, it is established
in~\cite{KKZ05} that, to obtain undecidability of two-variable
fragments, in the intuitionistic case, it suffices to use two binary
and infinitely many monadic predicate letters, while in the modal
case, it suffices to use only (infinitely many) monadic predicate
letters.

Two questions were raised in~\cite{KKZ05} concerning the languages
combining restrictions on the number of individual variables and
predicate letters: first, how many monadic predicate letters are
needed to obtain undecidability of the two-variable fragments in the
modal case, and second, whether it suffices to use monadic predicate
letters to obtain undecidability of the two-variable fragment in the
intuitionistic case.

In the present paper, we address both of the aforementioned questions.
First, we show that for two-variable fragments of most modal logics
considered in~\cite{KKZ05}, it suffices to use a single monadic
predicate letter to obtain undecidability.  Second, we show that the
positive fragment of quantified intuitionistic logic {\bf QInt} is
undecidable in the language with two variables and a single monadic
predicate letter.  We also show that the latter result holds true for
all logics in intervals $[{\bf QBL}, {\bf QKC}]$ and
$[{\bf QBL}, {\bf QFL}]$, where {\bf QKC} is the logic of the weak law
of the excluded middle and {\bf QBL} and {\bf QFL} are first-order
counterparts of Visser's basic and formal logics, respectively.

The paper is structured as follows. In section~\ref{sec:modal-logics},
we prove undecidability results about modal logics.  In
section~\ref{sec:int}, we do likewise for the intuitionistic and
related logics.  We conclude, in section~\ref{sec:conclusion}, by
discussing how our results can be applied in settings not considered
in this paper and pointing out some open questions following from our
work.

\section{Modal logics}
\label{sec:modal-logics}

In this section, we prove undecidability results about two-variable
fragments of quantified modal logics with a single monadic predicate
letter.  This is essentially achieved by adapting to the first-order
language of Halpern's technique~\cite{Halpern95} for establishing
complexity results for single-variable fragments of propositional
modal logics.

\subsection{Syntax and semantics}
\label{sec:syntax-semantics}

A (first-order) quantified modal language contains countably many
individual variables; countably many predicate letters of every arity;
Boolean connectives $\wedge$ and $\neg$; modal connective $\Box$; and
a quantifier $\forall$.  Formulas as well as the symbols $\vee$,
$\imp$, $\exists$, and $\Diamond$ are defined in the usual way. We
also use the following abbreviations:
$\Box^{+} \vp = \vp \con \Box \vp$ and
$\Diamond^{+} \vp = \vp \dis \Diamond \vp$.

A Kripke frame is a tuple $\frak{F} = \langle W,R\rangle$, where $W$
is a non-empty set (of worlds) and $R$ is a binary (accessibility)
relation on $W$.  A predicate Kripke frame is a tuple
$\frak{F}_D = \langle W,R, D\rangle$, where $\langle W,R\rangle$ is a
Kripke frame and $D$ is a function from $W$ into a set of non-empty
subsets of some set (the domain of~$\frak{F}_D$), satisfying the
condition that $wRw'$ implies $D(w) \subseteq D(w')$. We call the set
$D(w)$ the domain of $w$. We will also be interested in predicate
frames satisfying the condition that $wRw'$ implies $D(w) = D(w')$; we
refer to such frames as frames with constant domains.

A Kripke model is a tuple $\frak{M} = \langle W,R,D,I\rangle$, where
$\langle W,R, D\rangle$ is a predicate Kripke frame and $I$ is a
function assigning to a world $w\in W$ and an $n$-ary predicate letter
$P$ an $n$-ary relation $I(w,P)$ on $D(w)$.  We refer to $I$ as the
interpretation of predicate letters with respect to worlds in $W$.

An assignment in a model is a function $g$ associating with every
individual variable $x$ an element of the domain of the underlying
frame.

The truth of a formula $\varphi$ in a world $w$ of a model $\frak{M}$
under an assignment $g$ is inductively defined as follows:
\begin{itemize}
\item
$\frak{M},w\models^g P(x_1,\ldots,x_n)$ if $\langle g(x_1),\ldots,g(x_n)\rangle \in I(w,P)$;
\item $\frak{M},w\models^g\varphi_1\wedge\varphi_2$ if
  $\frak{M},w\models^g\varphi_1$ and $\frak{M},w\models^g\varphi_2$;
\item $\frak{M},w\models^g\neg\varphi_1$ if
  $\frak{M},w\not\models^g\varphi_1$;
\item $\frak{M},w\models^g\Box\vp_1$ if $wRw'$ implies
  $\frak{M},w'\models^g\varphi_1$, for every $w'\in W$;
\item $\frak{M},w\models^g\forall x\,\varphi_1$ if
  $\frak{M},w\models^{g'}\varphi_1$, for every assignment $g'$ such
  that $g'$ differs from $g$ in at most the value of $x$ and such that
  $g'(x)\in D(w)$.
\end{itemize}

Note that, given a Kripke model $\frak{M} = \langle W,R,D,I\rangle$
and $w \in W$, the tuple $\frak{M}_w = \langle D_w, I_w \rangle$,
where $D_w = D(w)$ and $I_w(P) = I(w,P)$, is a classical predicate
model.

We say that $\vp$ is true at world $w$ of model $\frak{M}$ and write
$\frak{M},w\models \vp$ if $\frak{M},w\models^g \vp$ holds for every
$g$ assigning to free variables of $\vp$ elements of $D(w)$.  We say
that $\vp$ is true in $\frak{M}$ and write $\frak{M} \models \vp$ if
$\frak{M},w\models \vp$ holds for every world $w$ of $\frak{M}$.  We
say that $\vp$ is true in predicate frame $\frak{F}_D$ and write
$\frak{F}_D \models \vp$ if $\vp$ is true in every model based on
$\frak{F}_D$. We say that $\vp$ is true in frame $\frak{F}$ and write
$\frak{F} \models \vp$ if $\vp$ is true in every predicate frame of
the form $\frak{F}_D$.  Finally, we say that a formula is true in a
class of frames if it is true in every frame from the class.

Let $\frak{M} = \langle W,R,D,I\rangle$ be a model, $w \in W$, and
$a_1, \ldots, a_n \in D(w)$.  Let $\vp(x_1, \ldots, x_n)$ be a formula
whose free variables are among $x_1, \ldots, x_n$.  We write
$\frak{M}, w \models \vp [a_1, \ldots, a_n]$ to mean
$\frak{M}, w \models^g \vp (x_1, \ldots, x_n)$, where
$g(x_1) = a_1, \ldots, g(x_n) = a_n$.

Given a propositional normal modal logic $L$, let ${\bf Q}L$ be
${\bf QCl} \oplus L$ where $\oplus$ is the operation of closure under
(predicate) substitution, modus ponens, generalization, and
necessitation.  Of particular interest to us are the quantified
counterparts {\bf QGL}, {\bf QGrz}, and {\bf QKTB} of propositinal
logics {\bf GL}, {\bf QGrz}, and {\bf KTB}.  We recall that {\bf GL}
is the logic of Kripke frames whose accessibility relation is
irreflexive, transitive, and contains no infinite ascending chains,
while {\bf Grz} is the logic of frames whose accessibility relation is
reflexive, transitive, antisymmetric, and does not contain infinite
ascending chains of pairwise distinct worlds (in other words, the
accessibility relation on the frames for {\bf Grz} is the reflexive
closure of the one on the frames for {\bf GL}).  We also recall that
{\bf QGL} and {\bf QGrz} are
Kripke-incomplete~\cite{Montagna84,Rybakov01}, but are valid on all
the frames for {\bf GL} and {\bf Grz}, respectively.  Thus, for
technical reasons---namely, to avoid being distracted with
Kripke-completeness---we define logics ${\bf QGL}^{sem}$ and
${\bf QGrz}^{sem}$ as the sets of quantified formulas true in all the
frames of {\bf GL} and {\bf Grz}, respectively.  What is important for
us is that ${\bf QGL} \subseteq {\bf QGL}^{sem}$ and
${\bf QGrz} \subseteq {\bf QGrz}^{sem}$.  Lastly, we recall that {\bf
  KTB} is the logic of Kripke frames whose accessibility relation is
reflexive and symmetric and that {\bf QKTB} is complete with respect
to this class of frames.

Given a logic $L$ and a closed formula $\vp$ in the language of $L$,
we say
that $\vp$ is $L$-satisfiable if $\neg \vp \not\in L$.  If $L$ is
complete with respect to a class $\frak{C}$ of frames,
$L$-satisfiability of $\vp$ amounts to $\vp$ being true at a world of
a model based on a frame in $\frak{C}$.

We now turn to addressing the question, raised in~\cite{KKZ05}, of how
many monadic predicate letters are needed in the language of
quantified modal logics to obtain undecidability of their two-variable
fragments.  Using suitable adaptations of a technique originally
proposed in~\cite{Halpern95}, and further refined in~\cite{ChRyb03},
\cite{RSh18a}, and \cite{RSh18b}, for propositinal languages, we show
that all sublogics of {\bf QGL}, {\bf QGrz}, and {\bf QKTB} are
undecidable in the language with a single monadic predicate letter.

\subsection{Sublogics of QGL and QGrz}
\label{sec:undecidability}

In this section, we prove that all sublogics of $\mathbf{QGL}$ and
$\mathbf{QGrz}$ are undecidable in languages with two individual
variables and a single monadic predicate letter.

In the proof, we rely on the undecidability result by Kontchakov,
Kurucz, and Zakharyaschev~\cite{KKZ05} concerning first-order modal
logics with two variables: we use their formulas encoding an
undecidable tiling problem as a basis for our reduction.  Even thought
their formulas are not suitable as are for our purposes, they can be
readily modified, as explained in the proof of Theorem~\ref{thr:Grz}
below, to be used in our argument.

We begin with a description of a satisfiability-preserving
transformation of formulas that we subsequently, in the proof of
Theorem~\ref{thr:Grz}, apply to a slightly modified version of the
formulas obtained by Kontchakov, Kurucz, and
Zakharyaschev~\cite{KKZ05}.

Let $\vp$ be a (closed) formula containing monadic predicate letters
$P_1, \ldots, P_n$. Let $P_{n+1}$ be a monadic predicate letter
distinct from $P_1, \ldots, P_n$ and let $B = \forall x\, P_{n+1}(x)$.
Define an embedding $\cdot'$ as follows:
\begin{center}
  \begin{tabular}{llll}
    ${P_i(x)}'$ & = & $P_i(x),$ $\mbox{~~where~} i \in \{1, \ldots, n \}$; \\
    $(\neg \phi)'$ & = & $\neg \phi'$; & \\
    $(\phi \con \psi)'$ & = & $\phi' \con \psi'$; & \\
    $(\forall x\, \phi)'$ & = & $\forall x\, \phi'$; & \\
    $(\Box \phi)'$ & = & $\Box (B \imp \phi')$. &
\end{tabular}
\end{center}

\begin{lemma}
  \label{lem:vp-B}
  Let $L \in \{{\bf QK}, {\bf QGL}^{sem}, {\bf QGrz}^{sem} \}$.  Then,
  $\vp$ is $L$-satisfiable if, and only if, $B \con \vp'$ is
  $L$-satisfiable.
\end{lemma}

\begin{proof}
  Assume that \sat{M}{w_0}{\vp}, for some \mmodel{M} based on a frame
  for $L$ and some $w_0$. Let $\mmodel{M'}$ be a model that extends
  $\mmodel{M}$ by setting $I(w, P_{n+1}) = D(w)$, for every $w \in W$.
  Then, \sat{M'}{w_0}{B \con \vp'}. Conversely, assume that
  \sat{M}{w_0}{B \con \vp'}, for some \mmodel{M} based on a frame for
  $L$. Let \mmodel{M'} be a submodel of \mmodel{M} with
  $W' = \{ w : \sat{M}{w}{B} \} $.  Then, \sat{M'}{w_0}{\vp}.  Note
  that, for every logic $L$ in the statement of the Lemma, \mmodel{M'}
  is based on a frame for $L$.
\end{proof}

\begin{remark}
  \label{rem:B}
  In view of the proof of Lemma~\ref{lem:vp-B}, if $B \con \vp'$ is
  satisfied in a model \mmodel{M}, we can assume, without a loss of
  generality, that $B$ is true in \mmodel{M}.
\end{remark}

Now, given a monadic predicate letter $P$, we inductively define the
following sequence of formulas:
\begin{center}
  \begin{tabular}{lll}
    $\delta_1(x)$ & = & $P (x) \con \Diamond (\neg P(x) \con \Diamond \Box^+
             P(x) )$; \\
    $\delta_{m+1}(x)$ & = & $P (x) \con \Diamond (\neg P(x) \con \Diamond \delta_{m}(x))$.
\end{tabular}
\end{center}
Using formulas from this sequence, define, for every
$k \in \{1, \ldots, n+1 \}$, the formula
$$
\alpha_k (x ) = \delta_{k}(x) \con \neg \delta_{k+1}(x) \con
\Diamond \Box^+ \neg P(x).
$$
%
We now define models associated with formulas
$\alpha_k (x)$.  For every $k \in \{1, \ldots, n+1 \}$, let
$\frak{F}_k = \langle W_k,R_k\rangle$ be a Kripke frame where
$W_k = \{w_k^0,\ldots,w_k^{2k}\}\cup\{w_k^\ast\}$ and $R_k$ is the
transitive closure of the relation
$\{\langle w_k^i,w_k^{i+1}\rangle : 0\leqslant i < 2k\}\cup\{\langle
w_k^0,w_k^\ast\rangle\}$.
For every such $k$, let $\frak{M}_k=\langle W_k,R_k,D,I\rangle$ be a
model with constant domains and let $a$ be an individual in the domain
of every $\frak{M}_k$ (other than that, the relationship between the
domains of $\frak{M}_k$s is immaterial at this point). We say that
$\frak{M}_k$ is $a$-suitable if
$$
\frak{M}_k,w\models P[a] ~ \Longleftrightarrow ~ w = w_k^{2i},
\mbox{~for $i\in\{0,\ldots,k\}$}.
$$

\begin{lemma}
\label{lem:model-M_k}
Let $a$ be an individual in the domain of the models
$\frak{M}_1,\ldots,\frak{M}_{n+1}$ and let
$\frak{M}_1,\ldots,\frak{M}_{n+1}$ be $a$-suitable. Then,
$$
\frak{M}_k,w\models \alpha_m[a] ~ \Longleftrightarrow ~
\mbox{$k=m$ and $w=w_k^0$}.
$$
\end{lemma}

\begin{proof}
Straightforward.
\end{proof}

\begin{remark}
  \label{rem:ref-closures}
  Notice that the statement of Lemma~\ref{lem:model-M_k} holds true if
  we replace the accessibility relations in
  $\frak{M}_1,\ldots,\frak{M}_{n+1}$ with their reflexive closures.
\end{remark}

Now, for every $\alpha_k(x)$, where $k \in \{1, \ldots, n+1 \}$,
define
$$
\beta_k(x) = \neg P(x)\con\Diamond \alpha_k(x).
$$
Let $\varphi^\ast$ be the result of replacing in $\varphi'$ of
$P_k(x)$ with $\beta_k(x)$, for every $k\in\{1,\ldots,n+1\}$.

Call a formula $\psi$ containing only monadic predicate letters {\em
  \mbox{$L$-suitable}} if either $\psi$ is not $L$-satisfiable or
$\psi$ is satisfiable in a model $\frak{M}$, based on a frame with
constant domains validating $L$, satisfying the downward inheritance
property for monadic letters: $\frak{M}\models \Diamond P(x)\to P(x)$,
for every monadic predicate letter~$P$ (we call such a model
\emph{$L$-suitable}).

\begin{lemma}
  \label{lem:vp-ast}
  Let $L \in \{{\bf QK}, {\bf QGL}^{sem}, {\bf QGrz}^{sem} \}$ and let
  $\vp$ be an $L$-suitable formula.  Then, $B \con \vp'$ is
  $L$-satisfiable if, and only if,
  $\forall x\, \beta_{n+1} (x) \con \vp^\ast$ is $L$-satisfiable.
\end{lemma}

\begin{proof}
  The right-to-left direction follows from the closure of $L$ under
  predicate substitution.  For the other direction, suppose that
  $B \con \vp'$ is {\bf QK}-satisfiable. Let
  $\frak{M}=\langle W,R,D,I\rangle$ be a model such that
  $\frak{M},w_0\models B \con \vp'$, for some $w_0\in W$. In view of
  Remark~\ref{rem:B}, we may assume, without a loss of generality,
  that $\frak{M}\models B$.

  For every $w\in W$ and every frame $\frak{F}_k$
  ($1 \leqslant k \leqslant n +1$), let
  $\frak{F}_k^w = \langle \{w\}\times W_k,R_k^w\rangle$ be an
  isomorphic copy of $\frak{F}_k$.  For every $w\in W$ and
  $k\in \{1,\ldots,n+1\}$, add $\{w\}\times W_k$ to $W$ to obtain the
  set $W^\ast$.  Define the relation $R^\ast$ on $W^\ast$ as follows:

  $$
  R^\ast = R\cup\bigcup 
  \big\{R_k^w \cup \{\langle w,(w,w_k^0)\rangle\} : w\in W, 1\leqslant k\leqslant n+1\big\}.
  $$
  Thus, for every $w\in W$, we make the roots of frames
  $\frak{F}_1^w,\ldots,\frak{F}_{n+1}^w$ accessible from~$w$.  Next,
  for every $u\in W^\ast$ let
$$
D^\ast(u) = \left\{
            \begin{array}{ll}
            D(u), & \mbox{if $u\in W$}, \\
            D(w), & \mbox{if $u\in \{w\}\times W_k$}.
            \end{array}
            \right.
$$
Finally, for every $u\in W^\ast$ and every $a\in D^\ast(u)$, let
$$
\langle a \rangle \in I^\ast(u,P) ~~\leftrightharpoons~~ \parbox[t]{250pt}{$u
  = (w, w_k^{2i})$, for some $w\in W$,
  $k\in\{1,\ldots,n+1\}$,
  and $i\in\{0,\ldots,k\}$;
  and $\frak{M},w\models P_k[a]$.}
$$
Let $\frak{M}^\ast = \langle W^\ast,R^\ast,D^\ast,I^\ast\rangle$.  It
immediately follows from Lemma~\ref{lem:model-M_k} that, for every
$w\in W$, every $a\in D(w)$, and every $k\in\{1,\ldots,n+1\}$,
$$
\frak{M},w\models P_k[a] ~ \Longleftrightarrow ~
\frak{M}^\ast,w\models \beta_k [a].
$$
We can then show that, for every $w\in W$, every subformula
$\psi(x_1,\ldots,x_m)$ of $\vp$, and every
$a_1,\ldots,a_m\in D(w)$,
$$
\frak{M},w\models \psi'[a_1, \ldots, a_m] ~ \Longleftrightarrow ~
\frak{M}^\ast,w\models  \psi^\ast[a_1, \ldots, a_m],
$$
where $\psi^\ast(x_1,\ldots,x_m)$ is obtained by substituting
$\beta_1(x),\ldots, \beta_{n+1}(x)$
for $P_1(x),\ldots,P_{n+1}(x)$ in $\psi'(x_1,\ldots,x_m)$.

The proof proceeds by induction. We only consider the modal case,
leaving the rest to the reader. In this case,
$\psi'(x_1,\ldots,x_m) = \Box(\forall x\, P_{n+1}(x) \to
\chi'(x_1,\ldots,x_m))$
and
$\psi^\ast(x_1,\ldots,x_m) = \Box(\forall x\, \beta_{n+1}(x) \to
\chi^\ast(x_1,\ldots,x_m))$.
If $\frak{M}^\ast,w\not\models \psi^\ast[a_1,\ldots,a_m]$, then there
exists $w'\in W^\ast$ with $wR^\ast w'$ such that
$\frak{M}^\ast,w'\models \forall x\, \beta_{n+1}(x)$ and
$\frak{M}^\ast,w'\not\models\chi^\ast[a_1,\ldots,a_m]$. The condition
$\frak{M}^\ast,w'\models \forall x\,\beta_{n+1}(x)$ guarantees that
$w'\in W$; therefore, we may apply the inductive hypothesis to
conclude that $\frak{M},w'\not\models\chi'[a_1,\ldots,a_m]$.  The
other direction is straightforward.

Thus,
$\frak{M}^\ast,w_0\models \forall x\, \beta_{n+1} (x) \con \vp^\ast$,
i.\,e., $\forall x\, \beta_{n+1} (x) \con \vp^\ast$ is {\bf
  QK}-satisfiable.

For ${\bf QGL}^{sem}$ and ${\bf QGrz}^{sem}$, the proof is similar.
The only difference is that, when defining the model
$\frak{M}^{\ast}$, instead of $R^\ast$ mentioned above, we take as the
accessibility relations its transitive, and its reflexive and
transitive, closure, respectively.  We only observe that, for atomic
formulas, the proof relies on $L$-suitability, for
$L\in \{{\bf QGL}^{sem},{\bf QGrz}^{sem}\}$, of $\vp$ and, hence, of
$B \wedge \vp'$: in the construction described above, we begin with an
$L$-suitable model for $B \wedge \vp'$.
\end{proof}

We can now prove our main result in this section.

\begin{theorem}
  \label{thr:Grz}
  Let $L$ be a logic such that
  ${\bf QK} \subseteq L\subseteq {\bf QGL}$ or
  ${\bf QK } \subseteq L\subseteq {\bf QGrz}$. Then, $L$ is
  undecidable in the language with two individual variables and a
  single monadic predicate letter.
\end{theorem}

\begin{proof}
    We first establish the following:

    \begin{sublemma}
      \label{lem:sublemma-1} Let
      $L\in \{{\bf QK}, {\bf QGL}^{sem},{\bf QGrz}^{sem}\}$.  Then,
      the problem of $L$-satisfiability of $L$-suitable formulas
      containing only two individual variables and only monadic
      predicate letters is undecidable.
    \end{sublemma}

    \begin{proof}  By reduction from an undecidable~\cite{Berger66}
      $\nat \times \nat$ tiling problem.

      Kontchakov, Kurucz, and Zakharyaschev~\cite{KKZ05} define, for a
      finite set $T$ of tile types (for a brief description of the
      tiling problem considered here, see
      Section~\ref{sec:undecidability-int}), the formula $\chi_T$, a
      conjunction of the formulas (we write $H(x,y)$ and $V(x,y)$ for,
      respectively, $\mathit{succ}_H(x,y)$ and $\mathit{succ}_V(x,y)$
      used by Kontchakov, Kurucz, and Zakharyaschev)
      $$ \begin{array}{l} \displaystyle \forall
           x\,\bigvee\limits_{t\in T}(P_t(x)\wedge\bigwedge\limits_{t'\ne
           t}P_{t'}(x));
           \smallskip\\
           \displaystyle \forall x\forall y\,(H(x,y)\to
           \bigwedge\limits_{right(t)\ne left(t')}\neg(P_t(x)\wedge
           P_{t'}(y)));
           \smallskip\\
           \displaystyle \forall x\forall y\,(V(x,y)\to
           \bigwedge\limits_{up(t)\ne down(t')}\neg(P_t(x)\wedge
           P_{t'}(y)));
           \smallskip\\
           \displaystyle \forall x\exists y\,H(x,y) \wedge \forall
           x\exists y\,V(x,y);
           \smallskip\\
           \displaystyle \forall x\forall y\,(H(x,y) \to \Box H(x,y));
           \smallskip\\
           \displaystyle \forall x\forall y\,(V(x,y) \to \Box V(x,y));
           \smallskip\\
           \displaystyle \forall x\forall y\,(\Diamond V(x,y) \to
           V(x,y));
           \smallskip\\
           \displaystyle \forall x\,\Diamond D(x);
           \smallskip\\
           \displaystyle \Box\forall x\forall y\,[V(x,y)\wedge \exists
           x\, (D(x)\wedge H(x,y)) \to \forall y\,(H(x,y)\to \forall
           x\,(D(x)\to V(y,x)))],  \end{array} $$
         and show that $\chi_T$ satisfies the condition
         $$
         \begin{array}{lclr} \mbox{$\chi_T$ is $L$-satisfiable} &
           \Longleftrightarrow & \mbox{$T$ tiles $\nat\times \nat$.}
           & ({\ast})
         \end{array}
         $$

         We effectively construct from $\chi_T$ an $L$-suitable
         formula $\chi^\star_T$ and show that the condition ($\ast$)
         remains satisfied if $\chi_T$ is replaced with
         $\chi^\star_T$.  Since the tiling of $\nat\times \nat$ by $T$
         is undecidable~\cite{Berger66}, the statement of the Sublemma
         follows.

         We first consider the case $L = \mathbf{QK}$. Let
         $\chi^\circ_T$ be the formula obtained from $\chi_T$ by
         substituting $\neg D(x)$ for $D(x)$.  It should be clear that
         $\chi^\circ_T$ is $\mathbf{QK}$-satisfiable if, and only if,
         $\chi_T$ is \mbox{$\mathbf{QK}$-satisfiable}.

         Replace in $\chi^\circ_T$ every occurrence of $\Box\psi$ by
         $\Box(\forall x\,Q(x)\to\psi)$ and substitute into so
         obtained formula the formulas
         $\Diamond(\neg Q_1^H(x)\wedge \neg Q_2^H(y))$ and
         $\Diamond(\neg Q_1^V(x)\wedge \neg Q_2^V(y))$ for,
         respectively, $H(x,y)$ and $V(x,y)$.  Denote the resultant
         formula by $\bar{\chi}^\circ_T$. Lastly, put
         $\chi^\star_T=\forall x\,Q(x)\wedge\bar{\chi}^\circ_T$.

         We first show that $\chi^\star_T$ is satisfiable if, and only
         if, $\chi_T^\circ$ is satisfiable (and so $\chi^\star_T$ is
         satisfiable if, and only if, $T$ tiles $\nat\times \nat$).

         Suppose $\chi_T^\circ$ is satisfiable.  Then, by ($\ast$), there
         exists a tiling $\tau\colon\nat\times\nat\to T$.  We use
         $\tau$ to define the model $\frak{M}$: let
         $$ \begin{array}{lcl}
                 W & = & \{w^\ast\} \cup \{w_{ij} : i,j\in \nat\} \cup \{w'_{ij} : i,j\in \nat\}; \\
                 R & = & \big( \{w^\ast\} \times \{w_{ij} : i,j\in
                         \nat\} \big)
                         \cup \big( \big( \{w^\ast\}\cup\{w_{ij} : i,j\in \nat\} \big)
                          \times \{w'_{ij} : i,j\in \nat\} \big); \\
                 D_w & = &
                            \nat\times\nat,
                            ~ \qquad \hfill\mbox{for every $w\in W'$;}
                 \\
                 I(w,Q) & = &
                                   \left\{ \begin{array}{rl}
                                             \varnothing & \mbox{if $w=w'_{ij}$,}\\
                                             \nat\times\nat
                                                                                   &
                                                                                     \mbox{if
                                                                                     $w\ne
                                                                                     w'_{ij}$;} \end{array}
                                                                                     \right. \\
                 I(w,Q_1^H) & = &
                                   \left\{ \begin{array}{rl}
                                             \nat\times\nat-\{\langle i,j\rangle\} & \mbox{if $w=w'_{ij}$,}\\
                                             \nat\times\nat
                                                                                   &
                                                                                     \mbox{if
                                                                                     $w\ne
                                                                                     w'_{ij}$;} \end{array}
                                                                                     \right. \\
                 I(w,Q_2^H) & = &
                                   \left\{ \begin{array}{rl}
                                             \nat\times\nat-\{\langle i+1,j\rangle\} & \mbox{if $w=w'_{ij}$,}\\
                                             \nat\times\nat
                                                                                     &
                                                                                       \mbox{if
                                                                                       $w\ne
                                                                                       w'_{ij}$;} \end{array}
                                                                                       \right. \\
                 I(w,Q_1^V) & = &
                                   \left\{ \begin{array}{rl}
                                             \nat\times\nat-\{\langle i,j\rangle\} & \mbox{if $w=w'_{ij}$,}\\
                                             \nat\times\nat
                                                                                   &
                                                                                     \mbox{if
                                                                                     $w\ne
                                                                                     w'_{ij}$;} \end{array}
                                                                                     \right. \\
                 I(w,Q_2^V) & = &
                                   \left\{ \begin{array}{rl}
                                             \nat\times\nat-\{\langle i,j+1\rangle\} & \mbox{if $w=w'_{ij}$,}\\
                                             \nat\times\nat
                                                                                     &
                                                                                       \mbox{if
                                                                                       $w\ne
                                                                                       w'_{ij}$;} \end{array}
                                                                                       \right. \\
                 I(w^\ast,D) & = & \nat\times\nat; \\
                 I(w_{ij},D) & = & \nat\times\nat - \{\langle i,j\rangle\}; \\
                 I(w'_{ij},D) & = & \varnothing; \\
                 I(w,P_t) & = & \{\langle
                                 i,j\rangle\in
                                 \nat\times\nat
                                 :
                                 \tau(i,j)=t\},
                                 ~ \hfill\qquad \mbox{for every $w\in W'$;} \\
                 \mathfrak{M} & = & \langle W,R,D,I\rangle. \\
               \end{array}
               $$
               The existence of $\tau$ implies that
               $\mathfrak{M}, w^\ast \models \chi^\star_T$.

               Conversely, suppose $\chi^\star_T$ is satisfiable, i.e.
               $\mathfrak{M}_0, w_0 \models \chi^\star_T$, for some
               model $\mathfrak{M}_0$ and world $w_0$.  If we remove
               from $\mathfrak{M}_0$ the worlds refuting
               $\forall x\,Q(x)$ and define the interpretation of
               letters $H$ and $V$ to be the sets of pairs satisfying,
               respectively,
               $\Diamond(\neg Q_1^H(x)\wedge \neg Q_2^H(y))$ and
               $\Diamond(\neg Q_1^V(x)\wedge \neg Q_2^V(y))$, then
               $\chi^\circ_T$ is satisfied at $w_0$ in the resultant
               model.

               It remains to show that $\chi^\star_T$ is
               $\mathbf{QK}$-suitable.  To see that it is, observe
               that the model $\mathfrak{M}$ defined above is
               $\mathbf{QK}$-suitable.

               Next, suppose
               $L \in \{ \mathbf{QGL}^{sem}, \mathbf{QGrz}^{sem}
               \}$. The argument given above for $\mathbf{QK}$ applies
               as is to $\mathbf{QGL}^{sem}$ since the model
               $\frak{M}$ defined above is a
               $\mathbf{QGL}^{sem}$-model. For $\mathbf{QGrz}^{sem}$,
               the argument is similar---the only difference is that,
               in defining a $\mathbf{QGrz}^{sem}$-suitable model
               satisfying $\chi^\star_T$, we take as the accessibility
               relation the reflexive closure of the relation $R$
               defined above.

               Since $\chi^\star_T$ contains only two individual
               variables and only monadic predicate letters, the
               statement of the sublemma follows.
             \end{proof}

             Now, let $F = \{ \neg \chi^\star_T :\, \mbox{$T$ tiles
               $\nat\times \nat$} \}$ (thus, $F$ contains only
             $L$-suitable formulas with two individual variables and
             only monadic predicate letters).  It follows from the
             proof of Sublemma~\ref{lem:sublemma-1} that
             $\mathbf{QK}\cap F = \mathbf{QGL}^{sem}\cap F =
             \mathbf{QGrz}^{sem}\cap F$ and that $\mathbf{QK}\cap F$
             is undecidable.


By Lemmas~\ref{lem:vp-B} and~\ref{lem:vp-ast}, for
$L\in\{\mathbf{QK},\mathbf{QGL}^{sem},\mathbf{QGrz}^{sem}\}$,
$$
\begin{array}{lcl}
  \neg \chi^\star_T \in L\cap F & \Longleftrightarrow & \forall x\, \beta_{n+1} (x) \imp
                                                        \neg (\neg \chi^\star_T)^\ast \in L,
  \end{array}
  $$
  which implies the statement of the theorem.
\end{proof}

\begin{corollary}
  \label{cor:Grz}
  {\bf QK}, {\bf QT}, {\bf QD}, {\bf QK4}, {\bf QS4}, {\bf QGL}, and
  {\bf QGrz} are undecidable in the language with two individual
  variables and a single monadic predicate letter.
\end{corollary}

\begin{remark}
  Theorem~\ref{thr:Grz} and Corollary~\ref{cor:Grz} hold true if we
  replace every logic $L$ mentioned in their statements with
  $L \oplus bf$, where
  $bf = \forall x\, \Box P(x) \imp \Box\, \forall x\, P(x)$; adding
  $bf$ to $L$ forces us to consider only predicate frames for $L$ with
  constant domains.
\end{remark}

We conclude this section by noticing that the results obtained herein
are quite tight.  In has been shown in~\cite{WZ01}, Theorem~5.1, that
for logics {\bf QK}, {\bf QT}, {\bf QK4}, and {\bf QS4}, adding---on
top of the restriction to at most two individual variables and a
single monadic predicate letter---the very slight restriction that
modal operators apply only to formulas with at most one free
individual variable results in decidable fragments.  As noticed
in~\cite{WZ01}, the same holds true for the other logics mentioned in
Corollary~\ref{cor:Grz}.

\subsection{Sublogics of {\bf QKTB}}

We now prove results similar to those established in the preceding
section for logics in the interval $[{\bf QK}, {\bf QKTB}]$, where
{\bf QKTB} is the predicate logic of reflexive and symmetric
frames. In so doing, we use an adaptation of a technique used
in~\cite{RSh18b} for proving results about computational complexity of
finite-variable fragments of sublogics of the propositional logic {\bf
  KTB}.



We proceed as in the previous section right up to the point where
formulas $\alpha_k$ and models $\mmodel{M}_k$ are defined.
Then, we define the formulas $\alpha_k$ as follows.  First, let
$$
\begin{array}{rclcrcl}
  \Box^0 \vp &=& \vp, & & \Box^{\leqslant 0} \vp &=& \vp, \\
  \Box^{n+1} \vp &=&  \Box \Box^n \vp, & & \Box^{\leqslant n + 1} \vp &=&
                                                                         \Box^{\leqslant n} \vp \con
                                                                         \Box^{n+1} \vp, \\
  \Diamond^n \vp &=& \neg \Box^n \neg \vp, && \Diamond^{\leqslant n} \vp &=& \neg \Box^{\leqslant n} \neg \vp.
\end{array}
$$
Next, inductively define, for every $k \in \{1, \ldots, n+1 \}$, the
following sequence of formulas:
\begin{center}
  \begin{tabular}{lll}
    $\delta(x)$ & = & $\Box^+ P(x)$; \\
    $\delta_k^k(x)$ & = & $\Box^{\leqslant k} \neg P(x)\con
                          \Diamond^{k+1} P(x)\con \Diamond^{k+2}\delta(x)$; \\
    $\delta_{i}^k(x)$ & = & $\Box^{\leqslant i}\neg P(x)\con
                            \Diamond^{i+1} P(x)\con
                            \Box \Diamond^{i+1} P(x)\con
                            \Diamond^{2i+3}\delta^k_{i+1}(x)$, where
                            $1\leqslant i < k$.
\end{tabular}
\end{center}
For notational convenience, let $\delta^1_2(x) = \delta(x)$.

Lastly, let, for every $k \in \{1, \ldots, n+1 \}$,
$$
\alpha_k (x) = P(x) \con \Diamond^2 \delta^k_1(x) \con \neg \Diamond^3
\delta^k_2(x). 
$$

Now we define models $\mmodel{M}_k$ associated with formulas
$\alpha_k$.  Given an individual $a$ and $k \in \{1, \ldots, n+1 \}$,
a model $\frak{M}_k$, whose domain contains $a$, looks as follows.
For brevity, we call some worlds $a$-worlds; if a world is not an
$a$-world, we call it an $\bar{a}$-world.  The model is a chain of
worlds whose root, $r_k$, is an $a$-world.  The root is part of a
pattern of worlds, described below, which is in turn succeeded by
three final $a$-worlds.  The pattern looks as follows: a single
$a$-world is followed by $2i + 1$ $\bar{a}$-worlds, for
$1 \leqslant i \leqslant k$.  Thus the chain looks as follows: the
root (an $a$-world), then three $\bar{a}$-worlds, then an $a$-world,
then five $\bar{a}$-worlds, then an $a$-world, \ldots, then an
$a$-world, then $2k + 1$ $\bar{a}$-worlds, then three $a$-worlds.  The
accessibility relation between the worlds of $\mmodel{M}_k$ is both
reflexive and symmetric.

We say that $\mmodel{M}_k$ is $a$-suitable if
$$
\frak{M}_k,w\models P[a] ~ \Longleftrightarrow ~ w \mbox{~is an
  $a$-world}.
$$
We can, then, prove the following analogue of
Lemma~\ref{lem:model-M_k}.
\begin{lemma}
  \label{lem:model-M_k-QKTB}
  Let $a$ be an individual in the domain of the models
$\frak{M}_1,\ldots,\frak{M}_{n+1}$ and let
$\frak{M}_1,\ldots,\frak{M}_{n+1}$ be $a$-suitable. Then,
$$
\frak{M}_k,w\models \alpha_m[a] ~ \Longleftrightarrow ~
\mbox{$k=m$ and $w=r_k$}.
$$
\end{lemma}
\begin{proof}
  Straightforward.
\end{proof}

Let
$$
\beta_k(x) = \neg P(x)\con \Box \Diamond P(x) \con \Diamond
\alpha_k(x),
$$
\noindent and let $\varphi^\ast$ be the result of replacing in
$\varphi'$ of $P_k(x)$ with $\beta_k(x)$, for every
$k\in\{1,\ldots,n+1\}$.


We can then prove the following analogue of Lemma~\ref{lem:vp-ast}:
\begin{lemma}
  \label{lem:vp-ast-QKTB}
  Let $L \in \{{\bf QK}, {\bf QKTB}\}$.  Then, $B \con \vp'$ is
  $L$-satisfiable if, and only if,
  $\forall x\, \beta_{n+1} (x) \con \vp^\ast$ is $L$-satisfiable.
\end{lemma}

\begin{proof}
  Analogous to the proof of Lemma~\ref{lem:vp-ast}, with the
  observation that the truth status of formulas $\alpha_k$ is not
  changed at the worlds of the models $\mmodel{M}_k$ once they get
  attached to the model $\frak{M}$ satisfying the formula
  $B \con \vp'$ to obtain the model $\frak{M}^\ast$ satisfying the
  formula $\forall x\, \beta_{n+1} (x) \con \vp^\ast$, even though
  their roots can now see the worlds of $\frak{M}$ due to the symmetry
  of the accessibility relation of $\frak{M}^\ast$. For a detailed
  argument showing that the truth status of formulas $\alpha_k$ in
  $\frak{M}^\ast$ at worlds from $\mmodel{M}_k$ is not affected, we
  refer the reader to the proof of Lemma 3.9 in~\cite{RSh18b}.
\end{proof}

Then, using an argument analogous to the one used in the proof of
Theorem~\ref{thr:Grz}, we obtain the following:

\begin{theorem}
  \label{thrm:QBT}
  Let $L$ be a logic such that
  ${\bf QK} \subseteq L\subseteq {\bf QKTB}$.  Then, $L$ is
  undecidable in the language with two individual variables and a
  single monadic predicate letter.
\end{theorem}

\begin{corollary}
  \label{cor:KTB}
  {\bf QKB} and {\bf QKTB} are undecidable in the language with two
  individual variables and a single monadic predicate letter.
\end{corollary}

\section{Intuitionistic and related logics}
\label{sec:int}

We now consider logics closely related to the quantified
intuitionistic logic {\bf QInt}.

\subsection{Syntax and semantics}
\label{sec:int-syntax-semantics}

The (first-order) quantified intuitionistic language contains
countably many individual variables; countably many predicate letters
of every arity; propositional constants $\bot$ (``falsehood'') and
$\top$ (``truth''); propositional connectives $\wedge$, $\vee$, and
$\to$; and quantifiers $\exists$ and $\forall$.  Formulas are defined
in the usual way; when parentheses are left out, $\wedge$ and $\vee$
are understood to bind tighter than $\imp$.  We also use the following
abbreviations: $\Box \vp = \truth \imp \vp$, $\Box^0 \vp = \vp$, and
$\Box^{n+1} \vp = \Box \Box^n \vp$.

A Kripke frame is a tuple $\frak{F} = \langle W,R\rangle$, where $W$
is a non-empty set (of worlds) and $R$ is a binary (accessibility)
relation on $W$ that is reflexive, anti-symmetric, and transitive.

A Kripke model $\frak{M} = \langle W,R,D,I\rangle$ is defined as in
the modal case, except that the interpretation function $I$ satisfies
the additional condition that $wRw'$ implies
$I(w,P) \subseteq I(w',P)$. An assignment is defined as in the modal
case.



The truth of a formula $\varphi$ in a world $w$ of a model $\frak{M}$
under an assignment $g$ is inductively defined as follows:

\begin{itemize}
\item
  $\frak{M},w\not\models^g \bot$;
\item
  $\frak{M},w \models^g \top$;
\item
$\frak{M},w\models^g P(x_1,\ldots,x_n)$ if $\langle g(x_1),\ldots,g(x_n)\rangle \in I(w,P)$;
\item $\frak{M},w\models^g\varphi_1\wedge\varphi_2$ if
  $\frak{M},w\models^g\varphi_1$ and $\frak{M},w\models^g\varphi_2$;
\item
$\frak{M},w\models^g\varphi_1\vee\varphi_2$ if
$\frak{M},w\models^g\varphi_1$ or $\frak{M},w\models^g\varphi_2$;
\item $\frak{M},w\models^g\varphi_1\to\varphi_2$ if $wRw'$ and
  $\frak{M},w'\models^g\varphi_1$ imply
  $\frak{M},w'\models^g\varphi_2$, for every $w'\in W$;
\item $\frak{M},w\models^g\exists x\,\varphi_1$ if
  $\frak{M},w\models^{g'}\varphi_1$, for some assignment $g'$ that
  differs from $g$ at most in the value of $x$ and such that
  $g'(x)\in D(w)$;
\item $\frak{M},w\models^g\forall x\,\varphi_1$ if
  $\frak{M},w'\models^{g'}\varphi_1$, for every $w'\in W$ such that
  $wRw'$ and every assignment $g'$ such that $g'$ differs from $g$ in
  at most the value of $x$ and such that $g'(x)\in D(w')$.
\end{itemize}

Truth in models, frames, and classes of frames is defined as in the
modal case.  {\bf QInt} is the set of formulas true in all frames.

We also consider some logics closely related to {\bf QInt}.  First,
{\bf QKC} is the quantified counterpart of the propositional logic
${\bf KC} = {\bf Int} + \neg p \dis \neg \neg p$.  Semantically, {\bf
  QKC} is characterized by the frames that satisfy the (convergence)
condition that $w R v_1$ and $w R v_2$ imply the existence of a world
$u$ such that $v_1 R u$ and $v_2 R u$.

Second, we consider quantified counterparts of Visser's basic
propositional logic {\bf BPL} and formal propositional logic {\bf
  FPL}~\cite{Visser81}: {\bf BPL} and {\bf FPL} are logics in the
intuitionistic language whose modal companions are {\bf K4} and {\bf
  GL}---that is, given the G\"{o}del's translation $t$ of the
intuitionistic language into the modal one (see, for example,
\cite{ChZ}, \S\ 3.9), ${\bf BPL} = t^{-1} ({\bf K4})$ and
${\bf FPL} = t^{-1} ({\bf GL})$.  Therefore, we define their
quantified counterparts as logics ${\bf QBL} = T^{-1} ({\bf QK4})$ and
${\bf QFL} = T^{-1} ({\bf QGL})$, where $T$ is the extension of $t$
with the following clauses:
$T( \exists x\, \vp ) = \exists x\, T(\vp)$; and
$T( \forall x_1 \ldots \forall x_n\, \vp ) = \Box \forall x_1 \ldots
\forall x_n\, T(\vp)$,
where $\vp$ does not begin with a universal quantifier.
To give the semantic account of {\bf QBL} and {\bf QFL}, we use Kripke
frames and models as defined for {\bf QInt}, except that the
accessibility relation is now only required to be anti-symmetric and
transitive.
The relation $\frak{M},w\models^g\vp$ is defined as in the
intuitionistic case, with the following modification for the universal
quantifiers:
\begin{itemize}
\item $\frak{M},w\models^g\forall x_1\ldots\forall x_n\,\varphi_1$,
  where $\varphi_1$ does not begin with a universal quantifier, if
  $\frak{M},w'\models^{g'}\varphi_1$, for every $w'\in W$ such that
  $wRw'$ and every assignment $g'$ such that $g'$ differs from $g$ in
  at most the values of $x_1,\ldots,x_n$ and such that
  $g'(x_1),\ldots,g'(x_n)\in D(w')$.
\end{itemize}
This clause is required to make, in the absence of reflexivity of the
accessibility relation, the formula $\forall x \forall y\, \vp$
equivalent to the formula $\forall y \forall x\, \vp$.  Then,
${\bf QBL}$ is sound (and complete) with respect to all such frames,
while ${\bf QFL}$ is sound (but not complete) with respect to the
subclass where the converse of the accessibility relation is
well-founded (i.\,e., with respect to the frames of the logic {\bf
  GL}).  For technical reasons, namely to avoid being distracted with
Kripke-completeness, we define the logic ${\bf QFL}^{sem}$ as the set
of formulas valid in all frames where the converse of the
accessibility relation is well-founded; all that matters to us is that
${\bf QFL} \subseteq {\bf QFL}^{sem}$.

\subsection{Undecidability results}
\label{sec:undecidability-int}

We now address the question, raised in~\cite{KKZ05}, of whether it
suffices to use only monadic predicate letters to obtain
undecidability of the two-variable fragment {\bf QInt}(2) of {\bf
  QInt}.  We show that, in fact, it suffices to use a \emph{single}
monadic predicate letter to obtain undecidability of {\bf QInt}(2).
We do so by suitably adapting the technique used in~\cite{Rybakov08}
to (polynomially) reduce satisfiability in propositional
intuitionistic logic ${\bf Int}$ to satisfiability in the fragment of
${\bf Int}$ with only two propositional variables.  As the technique
from~\cite{Rybakov08} requires that we work with positive formulas, we
first show that the {\it{positive}} monadic fragment of {\bf QInt}(2)
is undecidable.  We note here that transitioning from the
propositional language to the first-order one, we ``strengthen'' the
result from~\cite{Rybakov08} in the following sense: while in the
propositional case, (the positive fragment of) ${\bf Int}$ is
polynomially reducible to its two-variable subfragment, in the the
first-order case, we (polynomially) reduce (the positive fragment of)
{\bf QInt}(2) to its subfragment containing a single predicate
letter.\footnote{In light of~\cite{Nishimura60}, the reduction of
  ${\bf Int}$ to its single-variable fragment would imply that the
  complexity classes {\bf P} and {\bf PSPACE} are equivalent.}
Working with the positive fragment of {\bf QInt} also allows us to
extend our results to the interval $[{\bf QInt}, {\bf QKC}]$, as all
logics in this interval share the positive fragment.  Moreover, a
modification of this construction allows us to obtain analogous
results for logics in $[{\bf QBL}, {\bf QFL}]$.

It is proven in~\cite{KKZ05} that {\bf QInt}(2) is undecidable by
reducing the following undecidable tiling problem~\cite{Berger66} to
the complement of {\bf QInt}(2): given a finite set $T$ of tile types
that are tuples of colours
$t = \langle left(t), right(t), up(t), down(t) \rangle$, decide
whether $T$ tiles the grid $\nat \times \nat$ in the sense that there
exists a function $\tau : \nat \times \nat \to T$ such that, for every
$i, j \in \nat$, we have $up(\tau(i,j)) = down(\tau(i,j+1))$ and
$right(\tau(i,j)) = left(\tau(i+1,j))$. The results in this section
build on this proof.

We start off by proving that the positive fragment of {\bf QInt}(2)
containing two binary and an unlimited number of monadic predicate
letters, as well as two propositional variables, is undecidable. This
is achieved by eliminating the constant $\bot$ from the formulas used
in the proof of undecidability of {\bf QInt}(2) from~\cite{KKZ05}.
For most formulas from~\cite{KKZ05}, all we do is replace $\bot$ with
a propositional variable $q$.  The resultant formulas are listed below
for the reader's convenience; for ease of reference, we preserve the
numbering from~\cite{KKZ05}:
\begin{equation}
  \label{eq:1}
  \forall x\, \bigvee_{t \in T} ( P_t (x ) \con \bigwedge_{t' \ne t}
  (P_{t'} (x) \imp q)),
\end{equation}
\begin{equation}
  \label{eq:2}
  \bigwedge_{right(t) \ne left(t')} \forall x\, \forall y\, ( H(x,y) \con P_t
  (x) \con P_{t'} (y) \imp q),
\end{equation}
\begin{equation}
  \label{eq:3}
  \bigwedge_{up(t) \ne down(t')} \forall x\, \forall y\, ( V(x,y) \con P_t
  (x) \con P_{t'} (y) \imp q),
\end{equation}
\begin{equation}
  \label{eq:4}
  \forall x\, \exists y\, H (x, y) \con \forall x\, \exists y\, V (x,
  y),
\end{equation}
\begin{equation}
  \label{eq:5}
  \forall x\, \forall y\, ( V(x, y) \dis (V(x, y) \imp q)),
\end{equation}
\begin{equation}
  \label{eq:6}
  \forall x\, \forall y\, [ V(x,y) \con \exists x\, ( D(x) \con H( y,
  x) ) \imp \forall y\, ( H(x,y) \imp \forall x\, (D(x) \imp V( y, x )
  ))].
\end{equation}

\noindent Let $\psi^{+}_T$ be the conjunction of formulas (1) through (6).
Then, define
$$
\vp^{+}_T = \psi^{+}_T \imp ((\exists x\, (D(x) \imp q) \imp p) \imp
p),
$$
where $p$ is a propositional variable distinct from $q$.

\begin{lemma}
  $\vp^{+}_T \notin {\bf QInt}(2)$ if, and only if, $T$ tiles
  $\nat \times \nat$.
\end{lemma}

\begin{proof}
  The proof is a minor modification of the proof of Theorem 1
  from~\cite{KKZ05}, with $q$ essentially playing the role that
  ``falsehood'' plays in~\cite{KKZ05}.

  For the left to right direction, we observe that, given a model
  $\mmodel{M}$ and a world $w$ such that $\notsat{M}{w}{\vp^{+}_T}$,
  as well as an arbitrary $d \in D(w)$, there exists a world $u$ in
  $\mmodel{M}$ with $w R u$ such that $\sat{M}{u}{D[d]}$ and
  $\notsat{M}{u}{q}$.  This is a straightforward consequence of the
  fact that
  $\notsat{M}{w}{(\exists x\, (D(x) \imp q) \imp p) \imp p}$.  Given
  this, the argument from~\cite{KKZ05} applies.

  For the other direction, the model falsifying $\vp^{+}_T$ is
  different from the one used in~\cite{KKZ05} only in the evaluation
  of $p$ and $q$.  Thus, we use the same frame and interpretation of
  predicate letters as in~\cite{KKZ05}, and additionally make $q$
  false at every world of the model and make $p$ false at $w_0$ and
  true at every other world.
\end{proof}

Since $\vp^{+}_T$ is a positive formula, this immediately gives us the
following:

\begin{corollary}
  \label{cor:QInt2-pos}
  The positive fragment of ${\bf QInt}$ with two individual variables,
  two binary predicate letters, an unlimited number of monadic
  predicate letters, and two propositional variables is undecidable.
\end{corollary}

We now show how, drawing on an idea of Kripke's for modal
logics~\cite{Kripke62}, one can, in the positive fragment of {\bf
  QInt}, simulate binary predicate letters using monadic predicate
letters and propositional variables.  As this does not increase the
number of individual variables in a formula, it will allow us to
eliminate binary predicate letters from the formula $\vp^{+}_T$.

\begin{lemma}
  \label{lm:modelling-binary-predicates}
  Let $\chi$ be a positive formula in ${\bf QInt}$ containing an
  occurrence of a binary predicate letter $Q$, and let $Q_1$ and $Q_2$
  be monadic predicate letters, and $r$ and $s$ be propositional
  variables, not occurring in $\chi$.  Let $\chi'$ be the result of
  uniformly replacing every subformula of $\chi$ of the form $Q(x, y)$
  with $(Q_1(x) \con Q_2(y) \imp r) \dis s$.  Then,
  $\chi \in {\bf QInt}$ if, and only if, $\chi' \in {\bf QInt}$.
\end{lemma}

\begin{proof}
  The left-to-right direction follows from the closure of {\bf QInt}
  under substitution.  For the other direction, assume that there
  exist $\mmodel{M} = \langle W, R, D, I \rangle$ and $w_0 \in W$ such
  that $\notsat{M}{w_0}{\chi}$. We modify $\mmodel{M}$ to obtain a
  model $\mmodel{M}'$ falsifying $\chi'$ as follows.  For every
  $w \in W$ and every $a, b \in D(w)$ such that
  $\notsat{M}{w}{Q[a, b]}$, add to $W$ a world $w_{a, b}$ with
  $w R' w_{a, b}$ and let
  $$
  \begin{array}{lll}
    \mmodel{M}', w_{a, b} \not\models r; &  &   \\
    \mmodel{M}', w_{a, b} \models s; & &   \\
    \mmodel{M}', w_{a, b} \models Q_1 [d] & \leftrightharpoons & d
                                                                     = a; \\
    \mmodel{M}', w_{a, b} \models Q_2 [d] & \leftrightharpoons & d
    = b;
  \end{array}
  $$
  and let all the predicate letters different from $Q_1$ and $Q_2$ and
  occurring in $\chi'$ be universally true at every such world;
  likewise for propositional variables different from $r$ and
  $s$. Also, let $\mmodel{M}', w \not\models s$.

  Then we can show that $\sat{M}{w}{\theta}[a_1, \ldots, a_m]$ if, and
  only if, $\sat{M'}{w}{\theta'}[a_1, \ldots, a_m]$, for every
  subformula $\theta$ of $\chi$, every $w \in W$, and every
  $a_1, \ldots, a_m \in D(w)$, where $\theta'$ is the result of
  substituting in $\theta$ every occurrence of $Q(x, y)$ with
  $(Q_1(x) \con Q_2(y) \imp r) \dis s$. The proof is by induction on
  $\theta$.

  For the base case, first note that if \notsat{M}{w}{Q[a, b]}, then
  the presence in $\mmodel{M}'$ of the world $w_{a,b}$ guarantees that
  $\notsat{M'}{w}{(Q_1[a] \con Q_2[b] \imp r) \dis s}$; on the other
  hand, if \sat{M}{w}{Q[a, b]}, then
  $\sat{M'}{w}{(Q_1[a] \con Q_2[b] \imp r) \dis s}$, as
  \notsat{M}{u}{Q_1[a]} or \notsat{M}{u}{Q_2[b]}, for every $u$ with
  $w R' u$.

  The cases for $\theta = \theta_1 \dis \theta_2$,
  $\theta = \theta_1 \con \theta_2$, and $\theta = \exists x\, \theta_1$
  are straightforward.

  Let $\theta = \theta_1 \imp \theta_2$.  Assume that
  $\notsat{M'}{w}{\theta' [a_1, \ldots, a_m]}$.  Then,
  \sat{M'}{u}{\theta'_1[a_1, \ldots, a_m]} and
  \notsat{M'}{u}{\theta'_2[a_1, \ldots, a_m]}, for some $u \in W'$
  with $w R' u$. If we could apply the inductive hypothesis to $u$, we
  would be done.  To see that we can, notice that $\theta'_2$ is built
  out of atomic formulas and the formula
  $(Q_1(x) \con Q_2(y) \imp r) \dis s$, all of which are true under
  every assignment in every $w' \in W' - W$, using only $\con$,
  $\dis$, $\imp$, $\exists$, and $\forall$.  Therefore, $\theta'_2$ is
  true in every $w' \in W' - W$ under every assignment; hence,
  $u \in W$ and the inductive hypothesis is, therefore, applicable.
  Thus, \notsat{M}{w}{\theta[a_1, \ldots, a_m]}.  The other
  direction is straightforward.


  The case $\theta = \forall x\, \theta_1$ is similarly argued.
\end{proof}

Now, let $\xi^{+}_T$ be the result of replacing in $\vp^{+}_T$ of
$$
\begin{array}{rcl}
H( x, y ) & \mbox{with} & (H_1(x) \con H_2(y) \imp r_1) \dis s_1; \\
V( x, y ) & \mbox{with} & (V_1(x) \hfill \con \hfill V_2(y) \hfill
                          \imp r_2) \dis s_2.
\end{array}
$$
In view of Lemma~\ref{lm:modelling-binary-predicates},
$\xi^{+}_T \not\in {\bf QInt}$(2) if, and only if, $T$ tiles
$\nat \times \nat$.  As we can replace in $\xi^{+}_T$ a propositinal
variable such as $q$ with, say, $\exists x\, Q(x)$, we immediately
obtain the following:

\begin{theorem}
  \label{thr:int-2-positive}
  The positive monadic fragment of ${\bf QInt}$ with two individual
  variables is undecidable.
\end{theorem}

We now embed the positive monadic fragment of {\bf QInt}(2) into its
subfragment containing formulas with only one monadic predicate
letter, suitably adapting the technique from~\cite{Rybakov08}.  As
this embedding does not introduce any fresh variables, our main result
in this section immediately follows.

  We begin by defining the frame $\frak{F} = \langle W, R \rangle$ to
  be used in the construction of a refuting countermodel.  The frame
  $\frak{F}$, depicted in Figure~\ref{fig:F-2}, is made up of levels
  of worlds.  The top-most, unnumbered, level comprises $d_1$, $d_2$,
  $d'_2$ and $d_3$; level $0$ comprises $a^0_1$, $a^0_2$, $b^0_1$ and
  $b^0_2$; level $1$ comprises $a^1_1$, $a^1_2$, $a^1_3$, $b^1_1$,
  $b^1_2$ and $b^1_3$; the accessibility relation between these worlds
  is depicted by arrows (the arrows that can be inferred by
  reflexivity and transitivity are omitted).  The other levels are
  defined recursively.

  For each $k \geqslant 2$, level $k$ contains worlds $a^k_l$ and
  $b^k_l$, for every $l \in \{1, \ldots, s_k \}$, where $s_k$ is
  defined by recursion: $s_1 = 3$; $s_{k+1} = (s_k - 1)^2$.  To define
  instances of the accessibility relation between worlds of level
  $k+1$ and worlds of level of $k$, for each $k \geqslant 1$, take the
  lexicographic ordering of pairs $\langle i, j \rangle$, where
  $i, j \in \{2, \ldots, s_k\}$, and, provided $\langle i, j \rangle$
  is the $m$th pair in this ordering, put
  $$
  \begin{array}{lll}
    a^{k+1}_{m} R\, b^k_1,
    & a^{k+1}_{m} R\, a^k_i,
    & a^{k+1}_{m} R\, b^k_j,
      \smallskip\\
    b^{k+1}_{m} R\, a^k_1,
    & b^{k+1}_{m} R\, a^k_i,
    & b^{k+1}_{m} R\, b^k_j.
  \end{array}
  $$

  Let $\frak{N}_a = \langle W, R, D, I \rangle$ be an intuitionistic
  Kripke model with a constant domain $\mathcal{A}$ containing element
  $a$; we assume that $\mathcal{A}$ contains at least three
  elements---as we shall see, this assumption does not lead to a loss
  of generality.  We say that $\frak{N}_a$ is {\em $a$-suitable} if,
  for some $a' \in \mathcal{A} - \{a\}$,
  \begin{itemize}
  \item $I(d_2, P) = \mathcal{A} - \{a\}$;
  \item $I(d'_2, P) = \{a'\}$;
  \item $I(d_3, P) = \{a, a'\}$;
  \item $I(b_1^0, P) = \{a'\}$;
  \item $I(w, P) = \varnothing$, for every
    $w \in W - \{ d_2^{\phantom{1}}, d'_2,
    d_3^{\phantom{1}}, b_1^0\}$.
  \end{itemize}

  \begin{figure}
\centering
\begin{picture}(180,210)
\put(50,170){$\circ$}
\put(80,170){$\circ$}
\put(110,170){$\circ$}
\put(140,170){$\circ$}

\put(52.5,144){\vector(0,1){26.5}}
\put(112.5,144){\vector(0,1){26.5}}
\put(54,144){\vector(1,1){27}}
\put(111,144){\vector(-1,1){27}}
\put(84,144){\vector(1,1){27}}
\put(81,144){\vector(-1,1){27}}
\put(141,144){\vector(-1,1){27}}
\put(140.5,143.5){\vector(-2,1){56}}
\put(140.7,143.2){\vector(-3,1){86}}

\put(84.5,144){\vector(2,1){56}}
\put(114.5,144.5){\vector(1,1){27}}
\put(142.5,145){\vector(0,1){25}}

\put(38,172){$d_1$}
\put(67,172){$d_3$}
\put(118,172){$d_2$}
\put(148,172){$d'_2$}

\put(50,140){$\circ$}
\put(80,140){$\circ$}
\put(110,140){$\circ$}
\put(140,140){$\circ$}

\put(20,110){$\circ$}

\put(50,110){$\circ$}
\put(80,110){$\circ$}
\put(110,110){$\circ$}
\put(140,110){$\circ$}
\put(170,110){$\circ$}

\put(20,60){$\circ$}
\put(50,60){$\circ$}
\put(170,60){$\circ$}
\put(140,60){$\circ$}
\put(32,62){$\dots$}
\put(62,62){$\dots$}
\put(122,62){$\dots$}
\put(152,62){$\dots$}

\put(50,30){$\circ$}
\put(140,30){$\circ$}
\put(32,32){$\dots$}
\put(62,32){$\dots$}
\put(122,32){$\dots$}
\put(152,32){$\dots$}

\put(92,7){$\dots$}
\put(92,82){$\dots$}

\put(52.5,34){\vector(0,1){26.5}}
\put(142.5,34){\vector(0,1){26.5}}

\put(52.5,114){\vector(0,1){26.5}}
\put(112.5,114){\vector(0,1){26.5}}
\put(142.5,114){\vector(0,1){26.5}}

\put(51,33.5){\vector(-1,1){27}}
\put(144,33.5){\vector(1,1){27}}
\put(54,33.2){\vector(3,1){86}}
\put(140.7,33.2){\vector(-3,1){86}}

\put(24,114){\vector(1,1){27}}
\put(171,114){\vector(-1,1){27}}
\put(81,114){\vector(-1,1){27}}
\put(111,114){\vector(-1,1){27}}

\put(24.5,113.5){\vector(2,1){56}}
\put(54.5,113.5){\vector(2,1){56}}
\put(84.5,113.5){\vector(2,1){56}}

\put(170.5,113.5){\vector(-2,1){56}}
\put(140.5,113.5){\vector(-2,1){56}}

\put(38,142){$a_1^0$}
\put(64,142){$a_2^0$}
\put(96,142){$b_1^0$}
\put(148,142){$b_2^0$}


\put(25,105){$b_3^1$}
\put(55,105){$b_2^1$}
\put(85,105){$b_1^1$}
\put(115,105){$a_3^1$}
\put(145,105){$a_2^1$}
\put(175,105){$a_1^1$}

\put(25,68){$a_1^k$}
\put(55,68){$a_i^k$}
\put(145,68){$b_j^k$}
\put(175,68){$b_1^k$}

\put(55,22){$b_m^{k+1}$}
\put(145,22){$a_m^{k+1}$}

\end{picture}

\caption{Frame $\frak{F}$}
\label{fig:F-2}
\end{figure}

  We now define formulas of one free variable, $x$, so that each so
  defined formula $\psi(x)$ is associated with a world of an
  $a$-suitable model based on $\frak{F}$ (or a frame isomorphic to
  $\frak{F}$), in the following sense: for every $w \in W$, the
  relation $w \not\models \psi(a)$ holds if, and only if, $w$ sees the
  world associated with $\psi$.  For these formulas, we use notation
  making it clear which worlds they correspond to: the formula $D_i$,
  for $i \in \{1, 3\}$, corresponds to ${d}_i^{\phantom{i}}$, the
  formula $D_2$ corresponds to $d_2$ and $d'_2$, the formula $A^k_i$
  to ${a}^k_i$, and the formula $B^k_i$ to $b^k_i$.

  First, we define formulas associated with the worlds of the three
  top-most levels:
  $$
  \begin{array}{lclclcl}
    D_1 & = & \exists x\, P(x); & A_1^1(x) & = & A_1^0(x) \con A_2^0(x) \imp B_1^0(x) \dis B_2^0(x);   \\
    D_2(x) & = & \exists x\, P(x) \imp P(x); &  A_2^1(x) & = & A_1^0(x) \con
                                                             B_1^0(x) \imp
                                                             A_2^0(x) \dis
                                                             B_2^0(x); \\
    D_3(x) & = & P(x) \imp \forall x\, P(x); & A_3^1(x) & =
                                               & A_1^0(x) \con B_2^0(x)
                                                 \imp A_2^0(x) \dis
                                                 B_1^0(x);  \\
    A_1^0(x) & = & D_2(x) \imp D_1 \dis D_3(x); & B_1^1(x)& = & A_2^0(x) \con B_1^0(x) \imp A_1^0(x) \dis B_2^0(x); \\
    A_2^0(x) & = & D_3(x) \imp D_1 \dis D_2(x); & B_2^1(x)& = & A_2^0(x) \con B_2^0(x) \imp A_1^0(x) \dis B_1^0(x); \\
    B_1^0(x) & = & D_1 \imp D_2(x) \dis D_3(x); & B_3^1(x)& = & B_1^0(x) \con B_2^0(x) \imp A_1^0(x) \dis A_2^0(x). \\
    B_2^0(x) & = & A_1^0(x) \con A_2^0(x) \con B_1^0(x) \imp & & & & \\
        & & \hfill D_1 \dis D_2(x) \dis D_3(x); & & & &
  \end{array}
  $$
  We proceed by recursion. Assume formulas associated with the worlds
  of level $k$, where $k \geqslant 1$, have been defined.  Take the
  lexicographic ordering of pairs $\langle i, j \rangle$, where
  $i, j \in \{2, \ldots, s_k\}$, and, provided $\langle i, j \rangle$
  is the $m$th pair in this ordering, put
  $$
  \begin{array}{lcl}
    A_m^{k+1}(x) & = & A_1^k(x) \imp B_1^k(x) \dis A_i^k(x) \dis B_j^k(x); \\
    B_m^{k+1}(x) & = & B_1^k(x) \imp A_1^k(x) \dis A_i^k(x) \dis B_j^k(x).
  \end{array}
  $$

\begin{lemma}
  \label{lem:frame-F}
  Let $\frak{M} = \langle W, R, D, I \rangle $ be an $a$-suitable
  model and let $w \in W$.  Then,
  $$
\begin{array}{lclclcl}
  \notsat{M}{w}{A_m^k}[a] & \Longleftrightarrow & w R a_m^k & \mbox{and} & \notsat{M}{w}{B_m^k}[a] & \Longleftrightarrow & w R b_m^k.
\end{array}
$$
\end{lemma}

\begin{proof}
  Induction on $k$.
\end{proof}

Now, let $\vp$ be a positive formula containing monadic predicate
letters $P_1, \ldots, P_n$ (we may assume $n \geqslant 2$).  For each
$i \in \{1, \ldots, n\}$, define
$$
\alpha_i(x) = A_i^{n+1}(x) \dis B_i^{n+1}(x).
$$
Finally, let $\vp^\ast$ be the result of substituting, for every
$i \in \{1, \ldots, n\}$, of $\alpha_i(x)$ for $P_i(x)$ into $\vp$.
\begin{lemma}
  \label{lem:QInt-main-lemma}
  $\vp \in {\bf QInt}$ if, and only if, $\vp^\ast \in {\bf QInt}$.
\end{lemma}

\begin{proof}
  The right-to-left direction follows from the closure of { \bf QInt }
  under predicate substitution.  For the other direction, assume that
  \notsat{M_\vp}{w_0}{\vp} for some
  $\frak{M}_{\vp} = \langle W_{\vp}, R_{\vp}, D_{\vp}, I_{\vp}
  \rangle$
  and $w_0 \in W_{\vp}$. (We may assume without a loss of generality
  that the domain of $\frak{M}_{\vp}$ contains at least three
  individuals; we use this fact in the construction of $\frak{M}^\ast$
  below.) We need to construct a model $\frak{M}^\ast$ falsifying
  $\vp^\ast$ at some world.

    For every $w \in W_{\vp}$ and $a \in D_{\vp}(w)$, let
    $\frak{F}_a^{w} = \langle \{ \langle w, a \rangle\} \times
    W^{\phantom{0}}, R_a^{w} \rangle$ be an isomorphic copy of the
    frame $\frak{F}$ under the isomorphism
    $f\colon v \mapsto \langle w, a, v \rangle$.

  Let
  $$
  \begin{array}{lcl}
    W^\ast & = & \displaystyle
                 W_\vp \cup \big( \bigcup\limits_{w \in W_\vp} \big( \{w\}
                 \times  D_\vp(w) \big)  \times W \big).
  \end{array}
  $$
  Let $S$ be the smallest relation on $W^\ast$ such that
  \begin{itemize}
  \item $R_\vp \subseteq S$;
  \item
    $\bigcup\limits_{w \in W_\vp,\ a \in D_\vp (w)} R_a^{w} \subseteq
    S$;
  \item for every $w \in W_\vp$, every $v \in W^\ast - W_\vp$, every
    $a \in D_\vp (w)$ and every $i \in \{ 1, \ldots, n \}$,
    $$
    \begin{array}{lcll}
      w S v
        & \leftrightharpoons
        & \mbox{either}
        & \mbox{$v\in \{\langle w, a, a^{n+1}_i \rangle, \langle w, a, b^{n+1}_i \rangle \}$
          and $\frak{M}_\vp, w \not\models P_i [a]$}
        \\
        &
        & \mbox{or}
        & \mbox{$v\in\{\langle w, a, a^{n+1}_{n+1}\rangle,
          \langle w, a, b^{n+1}_{n+1}\rangle\}$,}
    \end{array}
    $$
  \end{itemize}
  and let $R^\ast$ be the reflexive transitive closure of $S$.

  Let $D^\ast (u) = D_\vp (u)$ if $u \in W_\vp$ and $D^\ast (u) =
  D_\vp (w)$ if $u = \langle w, a, v \rangle$, for some $w \in W_\vp$,
  $a \in D_\vp (w)$ and $v \in W$.

  Let $I^\ast$ be an interpretation function on
  $\langle W^\ast, R^\ast, D^\ast \rangle$ such that, for every $w \in
  W_\vp$ and every $a \in D_\vp (w)$,
  \begin{itemize}
  \item $I^\ast(\langle w, a, d_2 \rangle, P) = D_\vp (w) -
    \{a\}$;
  \item $I^\ast(\langle w, a, d'_2 \rangle, P) = \{a'\}$, where $a'$ is
    a fixed element of $D_\vp (w) - \{a\}$;
  \item $I^\ast(\langle w, a, d_3 \rangle, P) = \{a, a'\}$, where
    $a'$ is a fixed element of $D_\vp (w) - \{a\}$;
  \item $I^\ast(\langle w, a, b_1^0 \rangle, P) = \{a'\}$, where $a'$
    is a fixed element of $D_\vp (w) - \{a\}$;
  \item $I^\ast(u, P) = \varnothing$, for every
    $u \in W^\ast - \big( \{ d_2, d'_2, d_3, b^0_1 \} \times \{
    \langle w, c \rangle : w \in W_\vp, c \in D_\vp(w)\} \big)$.
  \end{itemize}
  Finally, let
  $\frak{M}^\ast = \langle W^\ast, R^\ast, D^\ast, I^\ast \rangle$.
  Evidently, $I^\ast$ satisfies the heredity condition; hence,
  $\frak{M}^\ast$ is an intuitionistic Kripke model.

  We can now show, by induction on $\psi$, that
  \sat{M_{\vp}}{w}{\psi[a_1, \ldots, a_m]}, if and only if,
  \sat{M^\ast}{w}{\psi^\ast[a_1, \ldots, a_m]}, for every
  $w \in W_{\vp}$, every $a_1, \ldots, a_m \in D^\ast(w)$, and every
  subformula $\psi$ of $\vp$.  We only consider the cases where $\psi$
  is atomic and where $\psi = \psi_1 \imp \psi_2$.

  We shall rely on the following sublemmas, whose proof we leave to
  the reader:

  \begin{sublemma}
    \label{lem:sublemma-2}
    For every $w \in W_\vp$ and every $a \in D_\vp (w)$,
    $$
    \begin{array}{lcl}
      \frak{M}^\ast, {w} \not\models {A_1^n}[a]
        & \mbox{and}
        & \frak{M}^\ast, {w} \not\models {B_1^n}[a].
    \end{array}
    $$
  \end{sublemma}

  \begin{sublemma}
    \label{lem:sublemma-3}
      Let $\frak{N}_a = \langle W, R, D, I \rangle $ be an
      $a$-suitable model with a constant domain $\mathcal{A}$ and let
      $b \in \mathcal{A} - \{a\}$.  Then, for every $w \in W$, every
      $k \geqslant 2$ and every $m \in \{ 1, \ldots, s_{k} \}$,
      $$
      \begin{array}{lcl}
        \frak{N}_a, {w} \models {A_m^k}[b]
        & \mbox{and}
        & \frak{N}_a, {w} \models {B_m^k}[b].
      \end{array}
      $$
    \end{sublemma}

    Observe that, for every $w \in W_\vp$ and every $a \in D_\vp(w)$,
    the submodel of $\frak{M}^\ast$ whose set of worlds is
    $\{ \langle w, a, v \rangle : v \in \{ \langle w, a \rangle \}
    \times W \}$ is a generated submodel of $\frak{M}^\ast$ and is
    $a$-suitable; hence, Lemma~\ref{lem:frame-F} and
    Sublemma~\ref{lem:sublemma-3} apply to such submodels of
    $\frak{M}^\ast$.

    We now proceed with induction.

  Let $\psi = P_i (x)$, and so
  $\psi^\ast = A_i^{n+1}(x) \dis B_i^{n+1}(x)$, for some $i \in \{1,
  \ldots, n\}$.

  Assume $\frak{M}_\vp, w \not\models P_i [a]$.  By definition of
  $\frak{M}^\ast$, both $w R^\ast \langle w, a, a^{n+1}_i \rangle$ and
  $w R^\ast \langle w, a, b^{n+1}_i \rangle$.  By
  Lemma~\ref{lem:frame-F}, both
  $\frak{M}^\ast, \langle w, a, a^{n+1}_i \rangle \not\models
  A_i^{n+1}[a]$ and
  $\frak{M}^\ast, \langle w, a, b^{n+1}_i \rangle \not\models
  B_i^{n+1}[a]$.  Hence, by heredity,
  $\frak{M}^\ast, w \not\models A_i^{n+1}[a]$ and
  $\frak{M}^\ast, w \not\models B_i^{n+1}[a]$.  Therefore,
  $\frak{M}^\ast, w \not\models A_i^{n+1}[a] \dis B_i^{n+1}[a]$.

  Conversely, assume
  $\frak{M}^\ast, w \not\models A_i^{n+1}[a] \dis B_i^{n+1}[a]$. Then,
  $\frak{M}^\ast, w \not\models A_i^{n+1}[a]$ and
  $\frak{M}^\ast, w \not\models B_i^{n+1}[a]$. Hence, there exist
  $u', u'' \in W^\ast$ and $i, j \in \{2, \ldots, s_n\}$ such that
  $u', u'' \in R^\ast (w)$ and
  $$
  \begin{array}{rrrr}
    u' \models A^n_1 [a];
    & u' \not\models B^n_1 [a];
    & u' \not\models A^n_i [a];
    & u' \not\models B^n_j [a]; \medskip \\
    u'' \models B^n_1 [a];
    & u'' \not\models A^n_1 [a];
    & u'' \not\models A^n_i [a];
    & u'' \not\models B^n_j [a].
  \end{array}
  \eqno{({\ast})}
  $$

  We show that $u' = \langle w, a, a^{n+1}_i \rangle$ and
  $u'' = \langle w, a, b^{n+1}_i \rangle$.

  Since $u' \models A^s_1 [a]$ and $u'' \models B^s_1 [a]$, by
  Sublemma~\ref{lem:sublemma-2}, $u', u'' \in W^\ast - W_\vp$.
  Therefore, from $u' \not\models B^n_1 (a)$ and
  $u'' \not\models A^n_1 [a]$ we obtain, by
  Sublemma~\ref{lem:sublemma-3}, that
  $u', u'' \in \{\langle w, a \rangle\} \times W$.  Hence, from
  (${\ast}$) we obtain by Lemma~\ref{lem:frame-F} that, for some
  $i, j \in \{ 2, \ldots, s_n \}$,
  $$
  \begin{array}{rrrr}
    \neg   u' R^\ast  \langle w, a, a^n_1 \rangle;
    & u' R^\ast \langle w, a, b^n_1 \rangle;
    & u' R^\ast \langle w, a, a^n_i \rangle;
    & u' R^\ast \langle w, a, b^n_j \rangle; \medskip \\
    \neg   u'' R^\ast  \langle w, a, b^n_1 \rangle;
    & u'' R^\ast \langle w, a, a^n_1 \rangle;
    & u'' R^\ast \langle w, a, a^n_i \rangle;
    & u'' R^\ast \langle w, a, b^n_j \rangle.
  \end{array}
  $$
  Now, in $\frak{F}$, and hence in $\frak{F}_w^a$, only worlds of
  level greater than $n$ see more than one world of level $n$.  Hence,
  $u'$ and $u''$ belong to a level greater than $n$.  Since every
  world of $\frak{F}_w^a$ of level greater that $n + 1$ sees $a^n_1$
  and $b^n_1$, by Lemma~\ref{lem:frame-F}, for every world $w$ of
  level greater than $n + 1$, both $w \not\models A^n_1[a]$ and
  $w \not\models B^n_1[a]$.  Therefore, $u'$ and $u''$ are worlds of
  level $n+1$.  Since $u' R^\ast u'$ and $u'' R^\ast u''$, in view of
  (${\ast}$), $\frak{M}^\ast, u' \not\models A_i^{n+1}[a]$ and
  $\frak{M}^\ast, u'' \not\models B_i^{n+1}[a]$.  Hence, by
  Lemma~\ref{lem:frame-F}, $u' R^\ast \langle w, a, a_i^{n+1} \rangle$
  and $u'' R'' \langle w, a, b_i^{n+1} \rangle$.  Therefore,
  $u' = \langle w, a, a^{n+1}_i \rangle$ and
  $u'' = \langle w, a, b^{n+1}_i \rangle$.

  Thus, $w R^\ast \langle a, a^{n+1}_i \rangle$ and
  $w R^\ast \langle w, a, b^{n+1}_i \rangle$. Hence,
  $\frak{M}^\ast, w \not\models P_i[a]$.

  Let $\psi = \psi_1 \imp \psi_2$.  Assume
  \notsat{M^\ast}{w}{\psi^\ast[a_1, \ldots, a_m]}.  Then,
  \sat{M^\ast}{u}{\psi_1^\ast[a_1, \ldots, a_m]} and
  \notsat{M^\ast}{u}{\psi_2^\ast[a_1, \ldots, a_m]}, for some
  $u \in W^\ast$ with $w R^\ast u$. If we could apply the inductive
  hypothesis to $u$, we would be done.  To see that we can, notice
  that $\psi_2^\ast$ is built out of formulas of the form
  $A_i^{n+1}(x) \dis B_i^{n+1}(x)$ using only $\con$, $\dis$, $\imp$,
  $\exists$, and $\forall$. As, in view of Lemma~\ref{lem:frame-F} and
  Sublemma~\ref{lem:sublemma-3}, formulas
  $A_i^{n+1}(x) \dis B_i^{n+1}(x)$ are true at every world in
  $W^\ast - W_{\vp}$ accessible from $W_{\vp}$, it follows that
  $u \in W_{\vp}$; the inductive hypothesis is, therefore, applicable.
  Thus, \notsat{M_{\vp}}{w}{\psi[a_1, \ldots, a_m]}.  The other
  direction is straightforward.

  We conclude that \notsat{M^\ast}{w_0}{\vp^\ast} and, thus,
  $\vp^\ast \notin {\bf QInt}$.
\end{proof}

As the construction of $\vp^\ast$ from $\vp$ did not introduce any
fresh individual variables, we have the following:

\begin{theorem}
  \label{thr:main}
  The positive fragment of ${\bf QInt}$ with two individual variables
  and a single predicate letter is undecidable.
\end{theorem}

We now extend the argument presented above to the logics in the
intervals $[{\bf QBL}, {\bf QKC}]$ and $[{\bf QBL}, {\bf QFL}]$.

First, to establish the undecidability of the two-variable fragments
of logics whose semantics might contain irreflexive worlds, we need to
slightly modify formulas (1) through (6) listed above.  Therefore, we
define $\psi^{\ast}_T$ to be the conjunction of $\psi^+_T$ and
following formula:
\begin{equation}
  \label{eq:visser}
  \tag{5a}
  \forall x\, \forall y\, ( H(x, y) \dis (H(x, y) \imp q)),
\end{equation}
and define
$$
\vp^{\ast}_T = \psi^{\ast}_T \imp [ (\exists x\, (D(x) \imp \Box^5 q ) \imp p ) \imp
\Box p].
$$
This enables us to prove, using the tiling problem described above,
that $T$ tiles $\nat \times \nat$ if and only if
$\vp^{\ast}_T \not\in L(2)$, where
$L \in \{{\bf QBL}, {\bf QFL}^{sem}\}$.  We leave the details of the
proof to the reader.  As the construction of $\vp^{\ast}_T$ is uniform
for both logics, it follows that the claim holds for every
$L \in [{\bf QBL}, {\bf QFL}^{sem}]$.  Notice that the same proof also
works for logics in $[{\bf QBL}, {\bf QKC}]$.  We simulate binary
predicate letters by monadic ones as for {\bf QInt}. We now show how
to simulate all monadic predicate letters with a single one.

For the interval $[{\bf QBL}, {\bf QKC}]$, notice that if we add to
the model $\frak{M}^\ast$ built in the proof of
Lemma~\ref{lem:QInt-main-lemma} a world $d$ accessible from every
element of $W^\ast$ and such that $I^\ast(d, P) = D(d)$, the resultant
model is a model of every logic in the interval
$[{\bf QBL}, {\bf QKC}]$.  Thus, we have the following:

\begin{theorem}
  \label{thr:qkc}
  Let $L$ be a logic in the interval $[{\bf QBL}, {\bf QKC}]$. Then,
  the positive fragment of $L$ with two individual variables and a
  single predicate letter is undecidable.
\end{theorem}

We next consider the interval $[{\bf QBL}, {\bf QFL}^{sem}]$.  In this
case, we need to make a more substantial modification to the frame
$\frak{F}$, as the semantics of ${\bf QFL}^{sem}$ prohibits the
existence of reflexive worlds.  We then proceed as follows. First, add
to $W$ worlds $\bar{d}_2$, $\bar{d}'_2$, and $\bar{d}_3$ with
$d_2 R \bar{d}_2$, $d'_2 R \bar{d}'_2$, and $d_3 R \bar{d}_3$.
Second, for every $k \geqslant 0$, do the following: for every world
$a^k_i$, add to $W$ the world $\bar{a}^k_i$ and, for every world
$b^k_i$, add to $W$ the world $\bar{b}^k_i$; also, let
$a^k_i R \bar{a}^k_i$ and $b^k_i R \bar{b}^k_i$, for every $k$ and
$i$.  Lastly, whenever in $\frak{F}$ we had $a^{k+1}_i R a_j^{k}$ or
$a^{k+1}_i R b_j^{k}$, let $\bar{a}^{k+1}_i R a_j^{k}$ and
$\bar{a}^{k+1}_i R b_j^{k}$; also, whenever we had
$b^{k+1}_i R a_j^{k}$ or $b^{k+1}_i R b_j^{k}$, let
$\bar{b}^{k+1}_i R a_j^{k}$ and $\bar{b}^{k+1}_i R b_j^{k}$.  We then
define $a$-suitable models so that $I(\bar{d}_2, P) = I(d_2, P)$,
$I(\bar{d}'_2, P) = I(d'_2, P)$, $I(\bar{d}_3, P) = I(d_3, P)$, and
for every $k$ and $i$, $I(\bar{a}_k^i, P) = I(a_k^i, P)$ and
$I(\bar{b}_k^i, P) = I(b_k^i, P)$. In essence, we created ``doubles''
for the worlds $d_2$, $d'_2$, $d_3$, $a_k^i$, and $b_k^i$, which serve
to evaluate formulas whose main connective is $\imp$ or $\forall$ at
the worlds whose doubles they are. Then, $a$-suitable models satisfy
the condition in the statement of Lemma~\ref{lem:frame-F}, and the
model $\frak{M}^\ast$ built in the proof of
Lemma~\ref{lem:QInt-main-lemma} becomes a model of every logic in
$[{\bf QBL}, {\bf QFL}^{sem}]$.  As
${\bf QFL} \subseteq {\bf QFL}^{sem}$, we have the following:

\begin{theorem}
  \label{thr:main}
  Let $L$ be a logic in the interval $[{\bf QBL}, {\bf QFL}]$.  Then,
  the positive fragment of $L$ with two individual variables and a
  single predicate letter is undecidable.
\end{theorem}

\begin{remark}
  Note that the results of this section hold true if we only consider
  frames with constant domains.
\end{remark}

\section{Discussion}
\label{sec:conclusion}


As already noticed, the results presented in the present paper
concerning sublogics of {\bf QGL} and {\bf QGrz} are quite tight: as
shown in~\cite{WZ01}, for all ``natural'' sublogics of {\bf QGL} and
{\bf QGrz}---including {\bf QK}, {\bf QT}, {\bf QD}, {\bf QK4}, {\bf
  QS4}, {\bf QGL}, and {\bf QGrz}---adding to the restriction to two
individual variables and a single monadic predicate letter considered
in section~\ref{sec:modal-logics} a minor restriction that the modal
operators only apply to formulas with at most one free variable,
results in decidable fragments of those logics.  It is not difficult
to notice that the results analogous to those obtained in
section~\ref{sec:modal-logics} can be obtained for quasi-normal logics
such as {\bf QS} (Solovay's logic) and Lewis's {\bf QS1}, {\bf QS2},
and {\bf QS3}~\cite{Feys65}.

A notable exception in our consideration of modal logics is {\bf QS5},
whose two-variable monadic fragment was shown to be undecidable
in~\cite{KKZ05}.  While it is not difficult to extend our results to
the multimodal version of {\bf QS5}---we need to modify the
construction used for sublogics {\bf QKTB} by substituting a
succession of two steps along distinct accessibility relations for a
single step along a single aceessibility relation in the frames of
$a$-suitable models---nor is it difficult to show, by encoding the
tiling problem used in~\cite{KKZ05}, that the two-variable fragment of
{\bf QS5} with two monadic predicate letters and infinitely many
propositional symbols is undecidable, the case of {\bf QS5} remains
elusive.  We conjecture that the fragment of {\bf QS5} with two
variables and a single monadic predicate letter is decidable.

On the other hand, it is relatively straightforward to show that the
two-variable fragment of {\bf QS5} with a single monadic predicate
letter and an infinite supply of individual variables is undecidable.
Indeed, let {\bf SIB} be the first-order theory of a symmetric
irreflexive binary relation $S$; it is well-known that {\bf SIB} is
undecidable~\cite{NerodeShore80, Kremer97}. We can then simulate
$S(x,y)$ as $\Box(\neg P(x)\vee\neg P(y))$ and show that, if a
quantified modal logic $L$ is valid on a frame containing a world
that can see infinitely many worlds, then $L$ is undecidable in
the language with a single monadic predicate letter (and infinitely
many individual variables). This observation covers all modal logics
considered in~\cite{KKZ05}, but not covered by the results of
section~\ref{sec:modal-logics}, including {\bf QS5}, {\bf QGL.3}, and
{\bf QGrz.3}.


By contrast, we can say nothing about superintuitionistic logics not
included in the interval $[{\bf QInt}, {\bf QKC}]$, as our proof
relies on the fact that we are working with the positive fragment of
those logics.  It is not essential to our proof that formulas
$A^k_i(x)$ and $B^k_i(x)$ be positive; however, by discarding their
positivity we would weaken, rather than strengthen, our results.

\section*{Acknowledgements}
\label{acknowledgements}

We are grateful to the anonymous referees for comments that helped to
significantly improve the presentation of the paper.

\end{document}